%% file: fossacs12b.tex
\documentclass[envcountsame]{llncs}%

\title{On the Complexity of the Equivalence Problem for Probabilistic Automata\thanks{Research supported by EPSRC grant EP/G069158.
The first author is supported by a postdoctoral fellowship of the German Academic Exchange Service (DAAD).}}

\usepackage{ifthen}
\newcommand{\techRep}{true} %% switch here between true and false
\newcommand{\iftechrep}{\ifthenelse{\equal{\techRep}{true}}}

\usepackage{amsmath}
\usepackage{amssymb}
\usepackage{graphicx}
\usepackage{tikz}
\usetikzlibrary{calc,automata}
\usepackage{wrapfig}

\author{Stefan Kiefer\inst{1} \and Andrzej S. Murawski\inst{2} \and
Jo\"el Ouaknine\inst{1} \and \\ Bj\"orn Wachter\inst{1} \and James Worrell\inst{1}}
\institute{Department of Computer Science, University of Oxford, UK\and Department of Computer Science, University of Leicester, UK}
\newcommand{\Q}{\mathbb{Q}}

\newcommand{\N}{\mathbb{N}}
\newcommand{\Z}{\mathbb{Z}}

\newcommand{\A}{\mathcal{A}}
\newcommand{\B}{\mathcal{B}}
\newcommand{\C}{\mathcal{C}}

\newcommand{\ns}[1]{n^{(#1)}}
\newcommand{\Ms}[1]{M^{(#1)}}

\newcommand\apex[1]{\textsc{apex}}

\newcommand{\valpha}{\boldsymbol{\alpha}}
\newcommand{\alphas}[1]{\valpha^{(#1)}}
\newcommand{\Gammas}[1]{\Gamma^{(#1)}}

\newcommand{\veta}{\boldsymbol{\eta}}
\newcommand{\etas}[1]{\veta^{(#1)}}

\newcommand{\vu}{\boldsymbol{u}}

\newcommand{\vv}{\boldsymbol{v}}

\newcommand{\vr}{\boldsymbol{r}}
\newcommand{\vx}{\boldsymbol{x}}

\newcommand{\vw}{\boldsymbol{w}}

\newenvironment{qcorollary}[1]{%
{\mbox{}\newline\noindent\bf Corollary #1.}
\begin{itshape}%
}{%
\end{itshape}%
}

\newenvironment{qproposition}[1]{%
{\mbox{}\newline\noindent\bf Proposition #1.}
\begin{itshape}%
}{%
\end{itshape}%
}
\begin{document}
%---------------------
\maketitle

\begin{abstract}
Deciding equivalence of probabilistic automata is a key problem for
establishing various behavioural and anonymity properties of
probabilistic systems.  In recent experiments a randomised equivalence
test based on polynomial identity testing outperformed deterministic
algorithms. In this paper we show that polynomial identity testing
yields efficient algorithms for various generalisations of the
equivalence problem. First, we provide a randomized \textbf{NC}
procedure that also outputs a counterexample trace in case of
inequivalence. Second, we consider equivalence of probabilistic cost
automata.  In these automata transitions are labelled with integer
costs and each word is associated with a distribution on costs,
corresponding to the cumulative costs of the accepting runs on that
word.  Two automata are equivalent if they induce the same cost
distributions on each input word.  We show that equivalence can be
checked in randomised polynomial time.  Finally we show that the
equivalence problem for probabilistic visibly pushdown automata is
logspace equivalent to the problem of whether a polynomial represented
by an arithmetic circuit is identically zero.
\end{abstract}

\section{Introduction}
Probabilistic automata were introduced by Michael Rabin~\cite{Rab63}
as an extension of deterministic finite automata.  Nowadays
probabilistic automata, together with associated notions of refinement
and equivalence, are widely used in automated verification and
learning.  Two probabilistic automata are said to be equivalent if
each word is accepted with the same probability by both automata.
Checking two probabilistic automata for equivalence has been shown
\iftechrep{to be}{}%
crucial for efficiently establishing various behavioural and
anonymity properties of probabilistic systems, and is the key
algorithmic problem underlying the \apex{} tool
\cite{MO05,LMOW08,Cav11}.

It was shown by Tzeng~\cite{Tzeng} that equivalence for probabilistic
automata is decidable in polynomial time.  By contrast, the natural
analog of language inclusion, that one automaton accepts each word
with probability at least as great as another automaton, is
undecidable~\cite{CL89} even for automata of fixed
dimension~\cite{BC03}.  It has been pointed out
in~\cite{CortesMohriRastogi} that the equivalence problem for
probabilistic automata can also be solved by reducing it to the
minimisation problem for weighted automata and applying an algorithm
of Sch\"{u}tzenberger~\cite{Schutzenberger}.

In~\cite{Cav11} we suggested a new \emph{randomised} algorithm which
is based on \emph{polynomial identity testing}.  In our experiments
\cite{Cav11} the randomised algorithm compared well with the
Sch\"{u}tzenberger-Tzeng procedure on a collection of benchmarks.  In
this paper we further explore the connection between polynomial
identity testing and the equivalence problem of probabilistic
automata.  We show that polynomial identity testing yields efficient
algorithms for various generalisations of the equivalence problem.

In Section~\ref{sec:rand} we give a new randomised \textbf{NC}
algorithm for deciding equivalence of probabilistic automata.  Recall
that \textbf{NC} is the subclass of \textbf{P} containing those
problems that can be solved in polylogarithmic parallel
time~\cite{GHR95} (see also Section~\ref{sec:prelim}).
Tzeng~\cite{Tzeng96} considers the path equivalence problem for
nondeterministic automata which asks, given nondeterministic automata
$\A$ and~$\B$, whether each word has the same number of accepting
paths in $\A$ as in~$\B$.  He gives a deterministic \textbf{NC}
algorithm for deciding path equivalence which can be straightforwardly
adapted to yield an \textbf{NC} algorithm for equivalence of
probabilistic automata.  Our new randomised algorithm has the same
parallel time complexity as Tzeng's algorithm, but it also outputs a
word on which the automata differ in case of inequivalence, which
Tzeng's algorithm cannot.  Our algorithm is based on the
\emph{Isolating Lemma}, which was used in~\cite{MVV87} to compute
perfect matchings in randomised \textbf{NC}.  The randomised algorithm
in~\cite{Cav11}, which relies on the Schwartz-Zippel lemma, can also
output a counterexample, exploiting the self-reducibility of the
equivalence problem---however it does not seem possible to use this
algorithm to compute counterexamples in~\textbf{NC}. Whether there is
a deterministic \textbf{NC} algorithm that outputs counterexamples in
case of inequivalence remains open.

In Section~\ref{sec:commutative} we consider equivalence of
probabilistic automata with one or more cost structures.  Costs (or
rewards, which can be considered as negative costs) are omnipresent in
probabilistic modelling for capturing quantitative effects of
probabilistic computations, such as consumption of time,
\mbox{(de-)}allocation of memory, energy usage, financial gains, etc.
We model each cost structure as an integer-valued \emph{counter}, and
annotate the transitions with counter changes.

In nondeterministic cost automata~\cite{ABK11,Krob94} the cost of a
word is the minimum of the costs of all accepting runs on that word.
In probabilistic cost automata we instead associate a probability
distribution over costs with each input word, representing the
probability that a run over that word has a given cost.  Whereas
equivalence for nondeterministic cost automata is
undecidable~\cite{ABK11,Krob94}, we show that equivalence of
probabilistic cost automata is decidable in randomised polynomial time
(and in deterministic polynomial time if the number of counters is fixed).  Our
proof of decidability, and the complexity bounds we obtain, involves a
combination of classical techniques of~\cite{Schutzenberger,Tzeng}
with basic ideas from polynomial identity testing.

%Since there are no \emph{a priori} bounds on the counters, they may
%induce an infinite state space.  In a finite-state probabilistic
%automaton, one may assume for the purposes of equivalence checking
%that there are no $\varepsilon$-transitions, since they can be
%eliminated using standard techniques. In the infinite-state case,
%however, an $\varepsilon$-elimination may result in an infinite number
%of new transitions, each one labelled with a different counter change.
%So it is necessary to consider cost automata with
%$\varepsilon$-transitions and to integrate $\varepsilon$-elimination
%within the equivalence check itself.  Technically, we consider,
%
%for each input word and each state, a rational \emph{generating
%  function} of the counter value.  We show that a finite number of
%input words and a finite number of input values for the generating
%functions suffice to distinguish two inequivalent weighted cost
%automata, leading to a polynomial-time algorithm for checking
%equivalence.

We present a case study in which costs are used to model the
computation time required by an RSA encryption algorithm, and show
that the vulnerability of the algorithm to timing attacks depends on
the \mbox{(in-)}equivalence of probabilistic cost automata.
In~\cite{Kocher96} two possible defenses against such timing leaks
were suggested.  We also analyse their effectiveness.

In Section~\ref{sec:vpa} we consider pushdown automata.  Probabilistic
pushdown automata are a natural model of recursive probabilistic
procedures, stochastic grammars and branching
processes~\cite{EY09,KEM06}.  The equivalence problem of deterministic
pushdown automata has been extensively
studied~\cite{Senizergues97,Stirling02}.  We study the equivalence
problem for \emph{probabilistic visibly pushdown automata
  (VPA)}~\cite{AM04}.  In a visibly pushdown automaton, whether the
stack is popped or pushed is determined by the input symbol being
read.

%One motivation to study this
%problem comes from the tool \apex{}~\cite{LMOW08}.  This tool verifies
%contextual equivalence of probabilistic programs by reduction to
%language equivalence of probabilistic automata.  The ability to decide
%equivalence of probabilistic VPA will allow us to extend the class of
%programs that \apex{} can handle, along the lines of~\cite{MW08}.
%

We show that the equivalence problem for probabilistic VPA is logspace
equivalent to \emph{Arithmetic Circuit Identity Testing (ACIT)}, which
is the problem of determining equivalence of polynomials presented via
arithmetic circuits~\cite{ABKM09}.  Several polynomial-time randomized
algorithms are known for \textbf{ACIT}, but it is a major open problem
whether it can be solved in polynomial time by a deterministic
algorithm.  The inter-reducibility of probabilistic VPA equivalence
and \textbf{ACIT} is reminiscent of the reduction of the positivity
problem for arithmetic circuits to the reachability problem for
recursive Markov chains~\cite{EY09}.  However in this case the
reduction is only in one direction---from circuits to recursive Markov
chains.

In the technical development below it is convenient to consider
$\Q$-weighted automata, which generalise probabilistic automata.  All
our results and examples are stated in terms of $\Q$-weighted
automata.  \iftechrep{Some proofs have been moved to an
  appendix.}{Missing proofs can be found in a technical
  report~\cite{KMOWW12:fossacs-report}.}

%\emph{Related Work.}  It is well-known that language equivalence for
%nondeterministic automata is \textbf{PSPACE}-complete; however Cho and
%Huynh~\cite{CH92} show that language equivalence for
%\emph{unambiguous} nondeterministic automata is in \textbf{NC}.
%

\section{Preliminaries} \label{sec:prelim}
\subsection{Complexity Classes}
Recall that \textbf{NC} is the subclass of \textbf{P} comprising those
problems considered efficiently parallelisable.  \textbf{NC} can be
defined via \emph{parallel random-access machines (PRAMs)}, which
consist of a set of processors communicating through a shared memory.
A problem is in \textbf{NC} if it can be solved in time $(\log
n)^{O(1)}$ (polylogarithmic time) on a PRAM with $n^{O(1)}$
(polynomially many) processors.  A more abstract definition of
\textbf{NC} is as the class of languages which have \textbf{L}-uniform
Boolean circuits of polylogarithmic depth and polynomial size.  More
specifically, denote by $\mathbf{NC}^k$ the class of languages which
have circuits of depth $O(\log^k n)$.  The complexity class
\textbf{RNC} consists of those languages with randomized \textbf{NC}
algorithms.  We have the following inclusions none of which is known
to be strict:
\[ \mathbf{NC}^1 \subseteq \mathbf{L} \subseteq \mathbf{NL}
   \subseteq \mathbf{NC}^2 \subseteq \mathbf{NC} \subseteq \mathbf{RNC} \subseteq \mathbf{P}
   \, .\]

Problems in \textbf{NC} include directed reachability, computing the
rank and determinant of an integer matrix, solving linear systems of
equations and the tree-isomorphism problem.  Problems that are
\textbf{P}-hard under logspace reductions include circuit value and
max-flow.  Such problems are not in \textbf{NC} unless
$\mathbf{P}=\textbf{NC}$.
  Problems in \textbf{RNC} include
matching in graphs and max flow in $0/1$-valued networks.  In both
cases these problems have resisted classification as either in
\textbf{NC} or \textbf{P}-hard.  See~\cite{GHR95} for more
details about \textbf{NC} and \textbf{RNC}.
\subsection{Sequence Spaces}
In this section we recall some results about spaces of
sequences~\cite{Rudin73}.

Given $s > 0$, define the following space of \emph{formal power series}:
\[ \ell_1(\mathbb{Z}^s) := \{ f : \mathbb{Z}^s \rightarrow \mathbb{R} :
\textstyle\sum_{\vv \in \mathbb{Z}^s} |f(\vv)| < \infty \} \ . \] Then
$\ell_1(\mathbb{Z}^s)$ is a complete vector space under the norm
$||f|| = \sum_{\vv \in \mathbb{Z}^s} |f(\vv)|$.  We can moreover endow
$\ell_1(\mathbb{Z}^s)$ with a Banach algebra structure with multiplication
\[ (f \mathrel{\ast} g)(\vv) := \sum_{\substack{\vu,\vw \in \mathbb{Z}^s \\ \vu+\vw=\vv}}
f(\vu) g(\vw) \, . \]

Given $n>0$ we also consider the space $\ell_1(\mathbb{Z}^s)^{n \times
  n}$ of $n \times n$ matrices with coefficients in
$\ell_1(\mathbb{Z}^s)$.  This is a complete normed linear space with
respect to the infinity matrix norm
\[ ||M|| := \max_{1 \leq i \leq n} \sum_{1 \leq j \leq n} ||M_{i,j}|| \, .\]
If we define matrix multiplication in the standard way, using the
algebra structure on $\ell_1(\mathbb{Z}^s)$, then $||M N || \leq ||M||
||N||$.  In particular, if $||M|| < 1$ then we can define a
Kleene-star operation by $M^* := (I-M)^{-1} = \sum_{k=0}^\infty M^k$.

\section{Weighted Automata}
\label{sec:rand}

To permit effective representation of automata we assume that all
transition probabilities are rational numbers.  In our technical
development it is convenient to work with \emph{$\Q$-weighted
  automata}~\cite{Schutzenberger}, which are a generalisation of
Rabin's probabilistic automata.

  A \emph{$\Q$-weighted automaton} $\A = (n, \Sigma, M, \valpha,
  \veta)$ consists of a positive integer $n \in \N$ representing the
  number of states, a finite alphabet~$\Sigma$, a map $M : \Sigma \to
  \Q^{n \times n}$ assigning a transition matrix to each alphabet
  symbol, an initial (row) vector $\valpha \in \Q^n$, and a final
  (column) vector $\veta \in \Q^n$.  We extend $M$ to $\Sigma^*$
  as the matrix product $M(\sigma_1 \ldots \sigma_k) := M(\sigma_1) \cdot \ldots \cdot
  M(\sigma_k)$.  The automaton~$\A$ assigns each word~$w$ a
  \emph{weight} $\A(w) \in \Q$, where $\A(w) := \valpha M(w) \veta$.
  An automaton~$\A$ is said to be \emph{zero} if $\A(w) = 0$ for all
  $w \in \Sigma^*$.  Two automata $\B, \C$ over the same
  alphabet~$\Sigma$ are said to be \emph{equivalent} if $\B(w) =
  \C(w)$ for all $w \in \Sigma^*$.
In the remainder of this section we present a randomised $\mathbf{NC}^2$ algorithm for
deciding equivalence of $\mathbb{Q}$-weighted automata and, in case of
inequivalence, outputting a counterexample.

Given two automata $\B,\C$ that are to be checked for equivalence, one
can compute an automaton~$\A$ with $\A(w) = \B(w) - \C(w)$ for all $w
\in \Sigma^*$.  Then $\A$ is zero if and only if $\B$ and~$\C$ are
equivalent.  Given $\B = (\ns{\B}, \Sigma, \Ms{\B}, \alphas{\B},
\etas{\B})$ and $\C = (\ns{\C}, \Sigma, \Ms{\C}, \alphas{\C},
\etas{\C})$, set $\A= (n, \Sigma, M, \valpha, \veta)$ with $n :=
\ns{\B} + \ns{\C}$ and\\[-1mm]
 \[
  M(\sigma) := \begin{pmatrix} \Ms{\B}(\sigma) & 0 \\ 0 & \Ms{\C}(\sigma) \end{pmatrix}\,, \qquad
  \valpha := (\alphas{\B}, -\alphas{\C})\,, \qquad
  \veta := \begin{pmatrix} \etas{\B} \\ \etas{\C} \end{pmatrix}\,.
 \]
This reduction allows us to focus on \emph{zeroness}, i.e., the
problem of determining whether a given  $\Q$-weighted automaton
$\mathcal{A}=(n,\Sigma,M,\valpha,\veta)$ is zero.  (Since
transition weights can be negative, zeroness is not the same as
emptiness of the underlying unweighted automaton.)
Note that a witness word $w \in \Sigma^*$ against zeroness of~$\A$
 is also a witness against the equivalence of $\B$ and~$\C$.
The following result from~\cite{Tzeng} is crucial.
\begin{proposition}
If $\A$ is not equal to the zero automaton then there exists a word $u \in
\Sigma^*$ of length at most $n-1$ such that $\A(u)\neq 0$.
\label{prop:short}
\end{proposition}

Our randomised $\mathbf{NC}^2$ procedure uses the Isolating Lemma of
Mulmuley, Vazirani and Vazirani~\cite{MVV87}.  We use this lemma in a
very similar way to~\cite{MVV87}, who are concerned with computing
maximum matchings in graphs in \textbf{RNC}.

\begin{lemma}
Let $\mathcal{F}$ be a family of subsets of a set
$\{x_1,\ldots,x_N\}$. Suppose that each element $x_i$ is assigned a
weight $w_i$ chosen independently and uniformly at random from
$\{1,\ldots,2N\}$.  Define the weight of $S \in \mathcal{F}$ to be
$\sum_{x_i \in S}w_i$.  Then the probability that there is a unique
minimum weight set in $\mathcal{F}$ is at least $1/2$.
\end{lemma}

We will apply the Isolating Lemma in conjunction with
Proposition~\ref{prop:short} to decide zeroness of a weighted
automaton $\A$.  Suppose $\A$ has $n$ states and alphabet $\Sigma$.
Given $\sigma \in \Sigma$ and $1 \leq i \leq n$, choose a weight
$w_{i,\sigma}$ independently and uniformly at random from the set
$\{1,\ldots,2|\Sigma|n\}$.  Define the weight of a word $u =
\sigma_1\ldots \sigma_k$, $k \leq n$, to be
$\mathrm{wt}(u) := \sum_{i=1}^k w_{i,\sigma_i}$.  (The reader
should not confuse this with the weight $\A(u)$ assigned to $u$ by the
automaton $\A$.)  Then we obtain a univariate polynomial $P$ from
automaton $\A$ as follows:
\[ P(x) = \sum_{k=0}^n \sum_{u \in \Sigma^k} \mathcal{A}(u) x^{\mathrm{wt}(u)} \, .\]

If $\A$ is equivalent to the zero automaton then clearly $P \equiv 0$.
On the other hand, if $\A$ is non-zero, then by
Proposition~\ref{prop:short} the set $\mathcal{F} = \{ u \in \Sigma^{\leq n} :
\A(u) \neq 0 \}$ is non-empty.  Thus there is a unique minimum-weight
word $u \in \mathcal{F}$ with probability at least $1/2$ by the
Isolating Lemma.  In this case $P$ contains the monomial
$x^{\mathrm{wt}(u)}$ with coefficient $\A(u)$ as its
smallest-degree monomial.  Thus $P \not\equiv 0$ with probability at
least $1/2$.

It remains to observe that from the formula
\[ P(x) = \valpha \left( \sum_{i=0}^n \prod_{j=1}^i \sum_{\sigma \in
    \Sigma} M(\sigma) x^{w_{j,\sigma}} \right) \veta\] and the fact
that iterated products of matrices of univariate polynomials can be
computed in $\mathbf{NC}^2$~\cite{Cook85} we obtain an $\mathbf{RNC}$
algorithm for determining zeroness of weighted automata.

It is straightforward to extend the above algorithm to obtain an
$\mathbf{RNC}$ procedure that not only decides zeroness of $\A$ but
also outputs a word $u$ such that $\A(u) \neq 0$ in case $\A$ is
non-zero. Assume that $\A$ is non-zero and that the random choice of
weights has isolated a unique minimum-weight word $u =
\sigma_1\ldots\sigma_k$ such that $\A(u) \neq 0$.  To determine
whether $\sigma \in \Sigma$ is the $i$-th letter of $u$ we can
increase the weight $w_{i,\sigma}$ by $1$ while leaving all other
weights unchanged and recompute the polynomial $P(x)$.  Then $\sigma$
is the $i$-th letter in $u$ if and only if the minimum-degree monomial
in $P$ changes.  All of these tests can be done independently,
yielding an $\mathbf{RNC}$ procedure.

\begin{theorem}   Given two weighted automata $\A$ and $\B$, there is an
\textbf{RNC} procedure that determines whether or not $\A$ and $\B$
are equivalent and that outputs a word $w$ with $\A(w) \neq \B(w)$ in
case $\A$ and $\B$ are inequivalent.
\label{thm:rnc-equiv}
\end{theorem}

\section{Weighted Cost Automata} \label{sec:commutative}
In this section we consider weighted automata with costs.  Each
transition has a cost, and the cumulative cost of a run is recorded in
a tuple of counters.  Transitions can also have negative costs, which
can be considered as rewards.  Note though that the counters do not
affect the control flow of the automata.  In Example~\ref{ex:rsa} we
use costs to record the passage of time in an encryption protocol.  We
explicitly include $\varepsilon$-transitions in our automata because
they are convenient for applications (cf.~Example~\ref{ex:geom}) and
we cannot rely on existing $\varepsilon$-elimination results in the
presence of costs.

Let $\Sigma$ be a finite alphabet not containing the symbol
$\varepsilon$.  A \emph{$\Q$-weighted cost automaton} is a tuple $\A =
(n,s,\Sigma,M,\valpha,\veta)$, where $n \in \N$ is the number of
states; $s \in \N$ is the number of counters; $M : \Sigma \cup
\{\varepsilon\} \to (C \to \Q)^{n \times n}$ is the transition
function, where $C=\{-1,0,1\}^s$ is the set of \emph{elementary cost
  vectors}; $\valpha \in \Q^n$ is an initial (row) vector; $\veta \in
\Q^n$ is a final (column) vector.  In this definition,
$M(\sigma)_{i,j}(\vv)$ represents the weight of a $\sigma$-transition
from state $i$ to $j$ with cost vector $\vv \in C$.  For the
semantics to be well-defined we assume that the total weight of all
outgoing $\varepsilon$-labelled transitions from any given state is
strictly less than $1$.

In order to define the semantics of weighted cost automata it is
convenient to use results on matrices of formal power series from
Section~\ref{sec:prelim}.  We can regard $M(\sigma)$ as an $n \times
n$ matrix whose entries are elements of the space
$\ell_1(\mathbb{Z}^s)$ of formal power series, where
$M(\sigma)_{i,j}(\vv)=0$ for $\vv \in \mathbb{Z}^s \setminus C$.  Our
convention on the total weight of $\varepsilon$-transitions is
equivalent to the requirement that $||M(\varepsilon)|| < 1$.  We next
extend $M$ to a map $M : \Sigma^* \to (\ell_1(\mathbb{Z}^s))^{n \times
  n}$ such that, given a word $w \in \Sigma^*$ and states $i,j$,
$M(w)_{i,j}(\vv)$ is the total weight of all $w$-labelled paths from
state $i$ to state $j$ with accumulated cost $\vv\in\mathbb{Z}^s$.
Given a word $w = \sigma_1\sigma_2 \ldots \sigma_m \in \Sigma^*$, we
define
\begin{gather}
M(w) := M(\varepsilon)^* M(\sigma_1) M(\varepsilon)^*
           \cdots M(\sigma_m) M(\varepsilon)^* \, .
\label{eq:prod-def}
\end{gather}
Finally, given $w \in \Sigma^*$ we define $\A(w) := \valpha \, M(w) \,
\veta$.  Then $\A(w)$ is an element of $\ell_1(\mathbb{Z}^s)$ such
that $\A(w)(\vv)$ gives the total weight of all accepting runs with
accumulated cost $\vv \in \mathbb{Z}^s$.

Let $\vx=(x_1,\ldots,x_s)$ be a vector of variables, one for each
counter.  Our equivalence algorithm is based on a representation of
$\A(w)$ as a rational function in $\vx$, following classical
ideas~\cite{N69}. Given $\vv \in \mathbb{Z}^s$ we denote by $\vx^{\vv}$
the monomial $x_1^{v_1}\cdots x_s^{v_s}$.  (Note that we allow
negative powers in monomials.)  We say that $f \in
\ell_1(\mathbb{Z}^s)$ has \emph{finite support} if $f(\vv)=0$ for all
but finitely many $\vv \in \mathbb{Z}^s$.  We identify such an $f$
with the polynomial $\sum_{\vv \in \mathbb{Z}^s} f(\vv)\vx^{\vv}$.  We
furthermore say that $f \in \ell_1(\mathbb{Z}^s)$ is \emph{rational}
if there exist $g,h : \mathbb{Z}^s \rightarrow \Q$ with finite support
such that $f \mathrel{\ast} h = g$.  We then identify $f$ with the
rational function \[{\sum_{\vv \in \mathbb{Z}^s} g(\vv)\vx^{\vv}}
\Big/ {\sum_{\vv \in \mathbb{Z}^s} h(\vv)\vx^{\vv}} \, . \] Note that
we can clear all negative exponents from the numerator and denominator
of such an expression.  Note also that sums and products of rational
functions correspond to sums and products in $\ell_1(\mathbb{Z}^s)$ in
the above representation.

\begin{proposition}
$M(w)$ can be represented as a matrix of rational functions in~$\vx$
  such that the numerator and denominator in each matrix entry have
  degrees at most $2n(s+1) \cdot |w|$.
\label{prop:rational}
\end{proposition}
\begin{proof}
From equation (\ref{eq:prod-def}) it suffices to show that
$M(\varepsilon)^*$ can be represented as a matrix of rational
functions with appropriate degree bounds.  Recall that
$M(\varepsilon)^*=(I-M(\varepsilon))^{-1}$, so it suffices to show
that $I-M(\varepsilon)$ (considered as a matrix of polynomials) has an
inverse that can be represented as a matrix of rational functions.
But the determinant formula yields that $\det(I-M(\varepsilon))$ is a
(non-zero) polynomial in $\vx$, thus the cofactor formula for
inverting matrices yields a representation of
$(I-M(\varepsilon))^{-1}$ as a matrix of rational functions in $\vx$
of degree at most $2ns$.  \qed
\end{proof}

An automaton~$\A$ is said to be \emph{zero} if $\A(w) \equiv 0$ for all $w
\in \Sigma^*$.  Two automata $\B, \C$ over the same alphabet~$\Sigma$
with the same number of counters are said to be \emph{equivalent} if
$\B(w) \equiv \C(w)$ for all $w \in \Sigma^*$.  As in
Section~\ref{sec:rand}, the equivalence problem can be reduced to the
zeroness problem, so we focus on the latter.

%Given two automata
%$\B,\C$ that are to be checked for equivalence, one can compute a
%``difference automaton''~$\A$ with $\A(w) = \B(w) - \C(w)$ for all $w
%\in \Sigma^*$.  Then $\A$ is zero if and only if $\B$ and~$\C$ are
%equivalent.  Given $\B = (\ns{\B},s, \Sigma, \Ms{\B}, \alphas{\B},
%\etas{\B})$ and $\C = (\ns{\C},s, \Sigma, \Ms{\C}, \alphas{\C},
%\etas{\C})$, we set $\A= (n,s, \Sigma, M, \valpha, \veta)$ with $n :=
%\ns{\B} + \ns{\C}$ and
% \[
%  M(\sigma) := \begin{pmatrix} \Ms{\B}(\sigma) & 0 \\ 0 & \Ms{\C}(\sigma) \end{pmatrix}\,, \quad \
%  \valpha := (\alphas{\B}, -\alphas{\C})\,, \quad \
%  \veta := \begin{pmatrix} \etas{\B} \\ \etas{\C} \end{pmatrix}\,.
% \]
%
%By Proposition~\ref{prop:rational} we have that $\A(w):= \valpha \,
%M(w) \, \veta$ is a rational function.  In particular $\A$ is zero if
%and only if $\A(w)$ is the zero as a rational function for all words
%$w \in \Sigma^*$.
The following proposition states that if there is a
word witnessing that $\A$ is non-zero, then there is a ``short'' such
word.  \newcommand{\stmtpropshortword}{ $\A$ is zero if and only if
  $\A(w)\equiv 0$ for all $w \in \Sigma^*$ of length at most $n-1$.  }
\begin{proposition}
\label{prop:shortword}
\stmtpropshortword
\end{proposition}
The proof, given in full \iftechrep{in
  Appendix~\ref{app:commutative}}{in~\cite{KMOWW12:fossacs-report}},
is similar to the linear algebra arguments
from~\cite{Schutzenberger,Tzeng}, but involves an additional twist.
The key idea is to substitute concrete values for the variables~$\vx$,
thereby transforming from the setting of infinite-dimensional vector
spaces of rational functions in~$\vx$ to a finite dimensional
setting where the arguments of~\cite{Schutzenberger,Tzeng} apply.

The decidability of zeroness (and hence equivalence) for weighted
cost automata follows immediately from
Proposition~\ref{prop:shortword}.  However, using polynomial identity
testing, we arrive at the following theorem.

\begin{theorem}
The equivalence problem for weighted cost automata is decidable in
randomised polynomial time.
\label{thm:counter}
\end{theorem}
\begin{proof}
We have already observed that the equivalence problem can be reduced
to the zeroness problem.  We now reduce the zeroness problem to polynomial
identity testing.

Given an automaton $\A = (n,s, \Sigma, M, \valpha,\veta)$, for each
word $w \in \Sigma^*$ of length at most $n$ we have a rational
expression $\A(w)$ in variables $\vx=(x_1,\ldots,x_s)$ which has degree
at most $d:=2n(s+1) \cdot n$ by Proposition~\ref{prop:rational}.

Now consider the set $R:=\{1,2,\ldots,2d\}$.  Suppose that we pick
$\vr \in R^s$ uniformly at random.  Denote by $\A(w)(\vr)$ the result
of substituting $\vr$ for $\vx$ in the rational expression $\A(w)$.
Clearly if $\A$ is a zero automaton then $\A(w)(\vr)=0$ for all $\vr$.
On the other hand, if $\A$ is non-zero then by
Proposition~\ref{prop:shortword} there exists a word $w\in\Sigma^*$ of
length at most $n$ such that $\A(w)\not\equiv 0$.  Since the degree of
the rational expression $A(w)$ is at most $d$ it follows from the
Schwartz-Zippel theorem~\cite{DL78,Schwartz80,Zippel79} that the
probability that $\A(w)(\vr)=0$ is at most $1/2$.

Thus our randomised procedure is to pick $\vr \in R^s$ uniformly at
random and to check whether $\A(w)(\vr)=0$ for some $w \in \Sigma^*$.
It remains to show how we can do this check in polynomial time.  To
achieve this we show that there is a $\Q$-weighted automaton $\B$ with
no counters such that $\A(w)(\vr)=\B(w)$ for all $w \in \Sigma^*$,
since we can then check $\B$ for zeroness using, e.g., Tzeng's
algorithm~\cite{Tzeng}.
The automaton $\B$ has the form $\B = (\ns{\B}, \Sigma, \Ms{\B},
\alphas{\B}, \etas{\B})$, where $\ns{\B}=n$, $\alphas{\B}=\valpha$,
$\etas{\B}=\veta$ and $\Ms{\B}(\sigma) =
\sum_{\vv\in\mathbb{Z}^s}M(\sigma)(\vv)\vr^{\vv}$ for all $\sigma \in
\Sigma$.
\qed
\end{proof}

\newcommand{\stmtcorcounter}{
For each fixed number of counters the equivalence problem for weighted
cost automata is decidable in deterministic polynomial time.
}
\begin{corollary} \label{cor:counter}
\stmtcorcounter
\end{corollary}
\iftechrep{See Appendix~\ref{app:commutative}}{See~\cite{KMOWW12:fossacs-report}} for a proof.

\begin{example}
We consider probabilistic programs that randomly increase and decrease
a single counter (initialised with~$0$) so that upon termination the
counter has a random value~$X \in \Z$.  The programs should be such
that $X$ is a random variable with $X = Y-Z$ where $Y$ and~$Z$ are
independent random variables with a geometric distribution with
parameters $p = 1/2$ and $p=1/3$, respectively.  (By that we mean that
$\Pr(Y=k) = (1-p)^k p$ for $k \in \{0, 1, \ldots\}$, and similarly for
$Z$.)  Figure~\ref{fig:ex-geom-dist-code} shows code in the syntax of
the \apex{} tool.
\begin{figure}
\begin{center}
\begin{tabular}{cc}
\begin{minipage}[t]{0.5\linewidth}
\begin{verbatim}
inc:com, dec:com |-
  var%2 flip;
  flip := 0;
  while (flip = 0) do {
    flip := coin[0:1/2,1:1/2];
    if (flip = 0) then {
     inc;
    };
  };
  flip := 0;
  while (flip = 0) do {
    flip := coin[0:2/3,1:1/3];
    if (flip = 0) then {
     dec;
    };
  }
:com
\end{verbatim}
\end{minipage}
&
\begin{minipage}[t]{0.5\linewidth}
\begin{verbatim}
inc:com, dec:com |-
  var%2 flip;
  flip := coin[0:1/2,1:1/2];
  if (flip = 0) then {
    while (flip = 0) do {
      flip := coin[0:1/2,1:1/2];
      if (flip = 0) then {
       inc;
      };
    };
  } else {
    flip := 0;
    while (flip = 0) do {
      dec;
      flip := coin[0:2/3,1:1/3];
    };
  }
:com
\end{verbatim}
\end{minipage}
%\\ ($\B$) & ($\C$)
\end{tabular}
\end{center}
\caption{Two \apex{} programs for producing a counter that is distributed as the difference between two geometrically distributed random variables.}
\label{fig:ex-geom-dist-code}
\end{figure}

The program on the left consecutively runs two while loops: it first
increments the counter according to a geometric distribution with
parameter~$1/2$ and then decrements the counter according to a
geometric distribution with parameter~$1/3$, so that the final counter
value is distributed as desired.  The program on the right is more
efficient in that it runs only one of two while loops, depending on a
single coin flip at the beginning.  It may not be obvious though that
the final counter value follows the same distribution as in the left
program.  We used the \apex{} tool to translate the programs to the
probabilistic cost automata $\B$ and~$\C$ shown in
Figure~\ref{fig:ex-geom-dist-aut}.
\begin{figure}
 \begin{tabular}{c@{\hspace{20mm}}c}
 \begin{tikzpicture}[shorten >=1pt,baseline=0]
  \node[state] (1) at (0,0) {$1, \frac16$};
  \node[state] (2) at (3,0) {$2, \frac13$};
%  \path[->] (1) edge [loop below] node[above] {$x$} node[below] {$1{-}p_x$} (1);
  \draw[->] (-0.8,0) -- (1);
  \draw[->] (1) .. controls +(250:1.4cm) and +(290:1.4cm) .. node[above] {$\varepsilon$} node[below,yshift=1mm] {$\frac12$ : inc} (1);
  \draw[->] (2) .. controls +(250:1.4cm) and +(290:1.4cm) .. node[above] {$\varepsilon$} node[below,yshift=1mm] {$\frac23$ : dec} (2);
  \draw[->] (1) to node[above] {$\varepsilon$} node[below] {$\frac13$ : dec} (2);
 \end{tikzpicture}
 &
 \begin{tikzpicture}[shorten >=1pt,baseline=0]
  \node[state] (1) at (0,0) {$1, \frac14$};
  \node[state] (2) at (3,1) {$2, \frac12$};
  \node[state] (3) at (3,-1) {$3, \frac13$};
  \draw[->] (-0.8,0) -- (1);
  \draw[->] (1) to node[above] {$\varepsilon$} node[below,sloped] {$\frac14$ : inc} (2);
  \draw[->] (2) .. controls +(110:1.4cm) and +(70:1.4cm) .. node[above] {$\varepsilon$} node[below,yshift=0.5mm] {$\frac12$ : inc} (2);
  \draw[->] (1) to node[above] {$\varepsilon$} node[below,sloped] {$\frac12$ : dec} (3);
  \draw[->] (3) .. controls +(250:1.4cm) and +(290:1.4cm) .. node[above] {$\varepsilon$} node[below,yshift=0.5mm] {$\frac23$ : dec} (3);
 \end{tikzpicture}
 \\
 ($\B$) & ($\C$)
 \end{tabular}
\caption{Automata produced from the code in
  Figure~\ref{fig:ex-geom-dist-code}.  The states are labelled with
  their number and their ``acceptance probability'' ($\veta$-weight).
  In both automata, state~1 is the only initial state ($\valpha_1 = 1$
  and $\valpha_i = 0$ for $i \ne 1$).  The transitions are labelled
  with the input symbol~$\varepsilon$, with a probability (weight) and
  a counter action (i.e.\ cost).}
\label{fig:ex-geom-dist-aut}
\end{figure}
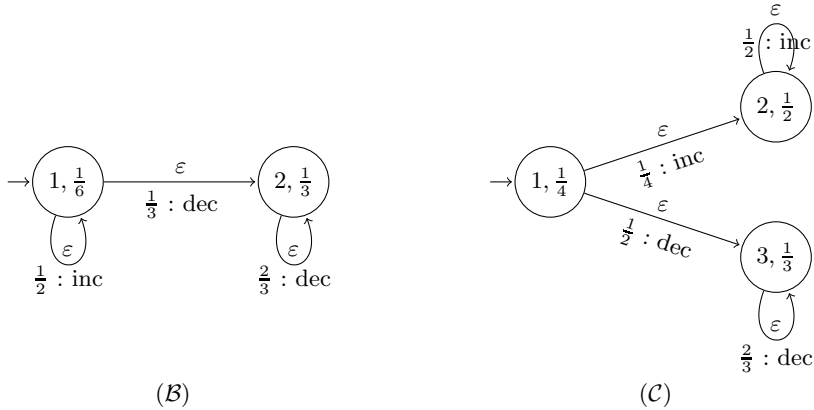
%
%We have $\B = (2,1,\emptyset,\Ms{\B},\alphas{\B},\etas{\B})$
%    and $\C = (3,1,\emptyset,\Ms{\C},\alphas{\C},\etas{\C})$.
Since the input alphabets are empty,
 it suffices to consider the input word~$\varepsilon$ when comparing $\B$ and~$\C$ for equivalence.
If we construct the difference automaton $\A = (5,1,\emptyset,M,\valpha,\veta)$
 and invert the matrix of polynomials $I-M(\varepsilon)$, we obtain
 \[
  \A(\varepsilon)(x) =
    \left(\frac{2}{x-2}, \frac{2}{(3x-2)(x-2)}, 1, \frac{-x}{2(x-2)}, \frac{3}{2(3x-2)}  \right) \veta \equiv 0\,,
 \]
 which proves equivalence of $\B$ and~$\C$.
Notice that the actual algorithm would not compute $\A(\varepsilon)(x)$ as a polynomial,
 but it would compute $\A(\varepsilon)(r)$ only for a few concrete values $r \in \Q$.
\qed
\label{ex:geom}
\end{example}

\begin{example}
RSA~\cite{RSA} is a widely-used cryptographic algorithm.  Popular
implementations of the RSA algorithm have been shown to be vulnerable
to timing attacks that reveal private keys~\cite{Kocher96,BrumleyB05}.
The preferred countermeasures are blinding techniques that randomise
certain aspects of the computation, which are described in,
e.g.,~\cite{Kocher96}.  We model the timing behaviour of the RSA
algorithm using probabilistic cost automata, where costs encode
time. These automata are produced by \apex{}, and are then used to
check for timing leaks with and without blinding.

At the heart of RSA decryption is a modular exponentiation,
which computes the value $m^d \mod N$ where $m\in \{0,\ldots,N-1\}$ is the encrypted message, $d\in \mathbb N$ is the private
decryption exponent and $N\in \mathbb N$ is a modulus. An attacker wants to find out $d$.
We model RSA decryption in \apex{} by implementing modular exponentiation
by iterative squaring (see Figure~\ref{fig:code-rsa}).
We consider the situation where the attacker is able
to control the message $m$, and tries to derive $d$ by observing the runtime distribution
over different messages $m$.
Following~\cite{Kocher96} we assume that the running time of multiplication depends on the operand values
(because a source-level multiplication typically corresponds to a cascade of processor-level multiplications).
By choosing the `right' input message $m$, an attacker can observe which private keys are most likely.

We consider two blinding techniques mentioned in Kocher~\cite{Kocher96}.
The first one is base blinding, i.e., the message is multiplied by $r^d$ before exponentiation where $d$ is a random number,
which gives a result that can be fixed by dividing by $r$ but makes it impossible for the attacker to control the
basis of the exponentiation.
The second one is exponent blinding, which adds a multiple of the group order $\varphi(N)$ of $\mathbb Z / N \mathbb Z$
to the exponent, which doesn't change the result of the exponentiation\footnote{Euler's totient function $\varphi$ satisfies $a^{\varphi(N)} \equiv 1 \mod N$ for all $a\in \mathbb Z$.}
but changes the timing behaviour.

Figure~\ref{fig:ex-rsa} shows the automaton
for $N=10$, and private key $0,1,0,1$ with message blinding enabled.
The \apex{} program is given in Figure~\ref{fig:code-rsa}.

We investigate the effectiveness of blinding.
Two private keys are indistinguishable if the resulting automata are equivalent.
The more keys are indistinguishable the safer the algorithm.
We analyse which private keys are identified by plain RSA, RSA with a blinded
message and RSA with blinded exponent.

For example, in plain RSA, the following keys $0,1,0,1$ and $1,0,0,1$
are indistinguishable, keys $0,1,1,0$ and $0,0,1,1$ are
indistinguishable with base blinding, lastly $1,0,0,1$ and $1,0,1,1$
are equivalent only with exponent blinding.  Overall 9 different keys
are distinguishable with plain RSA, 7 classes with base blinding and 4
classes with exponent blinding.
\begin{figure}
  \begin{center}
\begin{verbatim}
const N := 10;    // modulus
const Bits := 4 ; // number of bits of the key

m :int%N, inc:com |-
var%2 exponent[Bits] = [0,1,0,1];
com power(x:int%N) {
   var%N s := 1;
   var%N R;
   for(var%(Bits + 1) k; k < Bits; ++k) do {
      R:=s;
      if(exponent[k]) then {
         R := R*x;
         if(5<=R) then { inc; inc } else { inc }
      }
      s := R*R;
   }
}
var%N message := m*rand[N]; // blinding
power(message) : com
\end{verbatim}
\end{center}
\caption{\apex{} code for RSA.}
\label{fig:code-rsa}
\end{figure}
\begin{figure}
  \begin{center}
  \begin{tabular}{c}
  \includegraphics[scale=0.6, angle=270]{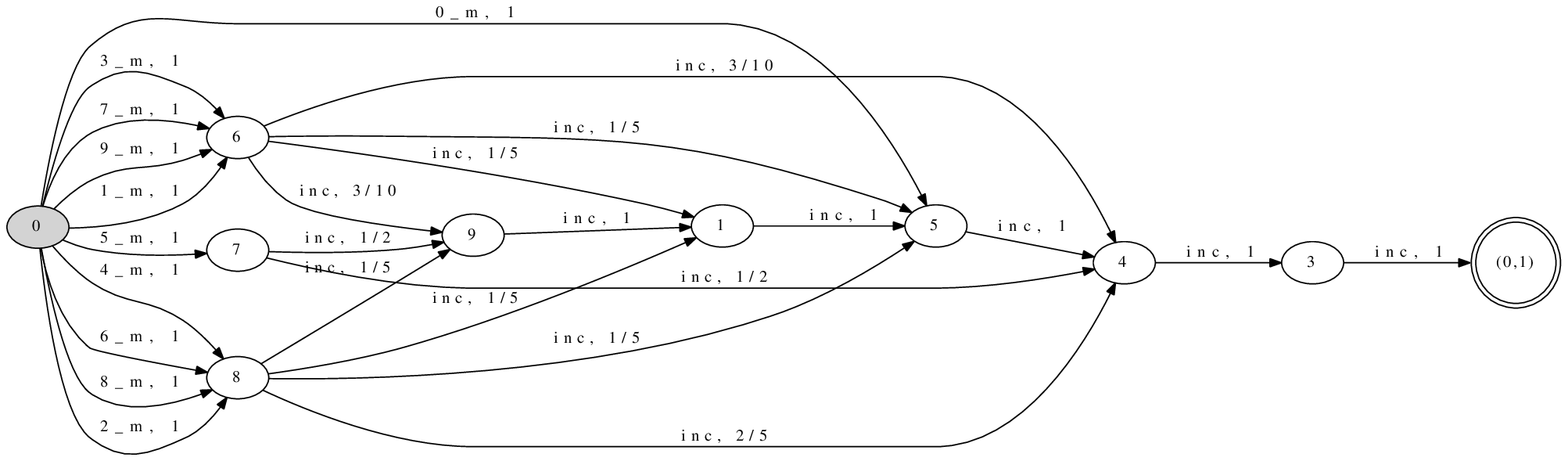}
\end{tabular}
\end{center}
\caption{Modeling RSA decryption with \apex{}.}
\label{fig:ex-rsa}
\end{figure}

\label{ex:rsa}
\end{example}

\section{Pushdown Automata and Arithmetic Circuits}
\label{sec:vpa}

In a visibly pushdown automaton~\cite{AM04} the stack operations are
determined by the input word.  Consequently VPA have a more tractable
language theory than ordinary pushdown automata.  The main result of
this section shows that the equivalence problem for weighted VPA
is logspace equivalent to the problem \textbf{ACIT} of determining
whether a polynomial represented by an arithmetic circuit is
identically zero.

A \emph{visibly pushdown alphabet} $\Sigma = \Sigma_c \cup \Sigma_r \cup
\Sigma_{\mathit{int}}$ consists of a finite set of \emph{calls}
$\Sigma_c$, a finite set of \emph{returns} $\Sigma_r$, and a finite
set of \emph{internal actions} $\Sigma_{\mathit{int}}$.
A visibly pushdown automaton over alphabet $\Sigma$ is
restricted so that it pushes onto the stack when it reads a call,
pops the stack when it reads a return, and leaves the stack untouched
when reading internal actions.  Due to this restriction visibly
pushdown automata only accept words in which calls and returns are
appropriately matched.  Define the set of \emph{well-matched words} to be
$\bigcup_{i \in \mathbb{N}}L_i$, where
$L_0  =  \Sigma_{\mathit{int}}+\{\varepsilon\}$ and
$L_{i+1} = \Sigma_c L_i \Sigma_r + L_i L_i$.

A \emph{$\Q$-weighted visibly pushdown automaton} on alphabet $\Sigma$
is a tuple $\mathcal{A}=(n,\valpha,\veta,\Gamma,M)$, where $n$ is the
number of \emph{states}, $\valpha$ is an $n$-dimensional
\emph{initial} (row) vector, $\veta$ is an $n$-dimensional
\emph{final} (column) vector, $\Gamma$ is a finite \emph{stack
  alphabet}, and $M= (M_c,M_r,M_{\mathit{int}})$ is a tuple of
\emph{matrix-valued transition functions} with types $M_c : \Sigma_c
\times \Gamma \to \Q^{n \times n}$, $M_r : \Sigma_r \times \Gamma \to
\Q^{n \times n}$ and $M_{\mathit{int}} : \Sigma_{\mathit{int}} \to
\Q^{n \times n}$.  If $a \in \Sigma_c$ and $\gamma \in \Gamma$ then
$M_c(a,\gamma)_{i,j}$ gives the weight of an $a$-labelled transition
from state $i$ to state $j$ that pushes $\gamma$ on the stack.  If $a
\in \Sigma_r$ and $\gamma \in \Gamma$ then $M_r(a,\gamma)_{i,j}$ gives
the weight of an $a$-labelled transition from state $i$ to $j$ that
pops $\gamma$ from the stack.

For each well-matched word $u \in \Sigma^*$ we define an $n \times n$
rational matrix $\Ms{\A}(u)$ whose $(i,j)$-th entry denotes the total
weight of all paths from state $i$ to state $j$ along input $u$.  The
definition of $\Ms{\A}(u)$ follows the inductive definition of
well-matched words.  The base cases are $\Ms{\A}(\varepsilon) = I$ and
$\Ms{\A}(a)_{i,j}= M_{\mathit{int}}(a)_{i,j}$.  The inductive cases are
\begin{eqnarray*}
\Ms{\A}(uv) & = & \Ms{\A}(u) \cdot \Ms{\A}(v)\\ \Ms{\A}(aub) & = & \sum_{\gamma \in \Gamma}
M_c(a,\gamma) \cdot \Ms{\A}(u) \cdot M_r(b,\gamma) \, ,
\end{eqnarray*}
for $a \in \Sigma_c$, $b \in \Sigma_r$.

The weight assigned by $\A$ to a well-matched word $w$ is defined to
be $\A(w):=\valpha \Ms{\A}(u) \veta$.  We say that two weighted VPA
$\A$ and $\B$ are \emph{equivalent} if for each well-matched word $w$
we have $\A(w)=\B(w)$.

An \emph{arithmetic circuit} is a finite directed acyclic multigraph
whose vertices, called \emph{gates}, have indegree $0$ or $2$.
Vertices of indegree $0$ are called \emph{input gates} and are
labelled with a constant $0$ or $1$, or a variable from the set $\{x_i
: i \in \mathbb{N}\}$.  Vertices of indegree $2$ are called
\emph{internal gates} and are labelled with one of the arithmetic
operations $+$, $*$ or $-$.  We assume that there is a unique gate
with outdegree $0$ called the \emph{output}.  Note that $C$ is a
multigraph, so there can be two edges between a pair of gates, i.e.,
both inputs to a given gate can lead from the same source.  We call a
circuit \emph{variable-free} if all inputs gates are labelled $0$ or
$1$.

The \emph{Arithmetic Circuit Identity Testing (\textbf{ACIT})} problem
asks whether the output of a given circuit is equal to the zero
polynomial.  \textbf{ACIT} is known to be in \textbf{coRP} but it
remains open whether there is a polynomial or even sub-exponential
algorithm for this problem~\cite{ABKM09}.  Utilising the fact
that a variable-free arithmetic circuit of size $O(n)$ can compute
$2^{2^n}$, Allender \emph{et al.}~\cite{ABKM09} give a logspace
reduction of the general \textbf{ACIT} problem to the special case of
variable-free circuits.  Henceforth we assume without loss of
generality that all circuits are variable-free.  Furthermore we recall
that \textbf{ACIT} can be reformulated as the problem of deciding
whether two variable-free circuits using only the arithmetic
operations $+$ and $*$ compute the same number~\cite{ABKM09}.

The proof of the following proposition is given
\iftechrep{in Appendix~\ref{theAppendix}.}{in~\cite{KMOWW12:fossacs-report}.}
\newcommand{\stmtpropACITtoVPA}{
$\mathbf{ACIT}$ is logspace reducible to the
equivalence problem for weighted visibly pushdown automata.
}
\begin{proposition} \label{prop:ACIT-to-VPA}
\stmtpropACITtoVPA
\end{proposition}

In the remainder of this section we give a converse reduction: from
equivalence of weighted VPA to \textbf{ACIT}.  The following
result gives a decision procedure for the equivalence of two
weighted VPA $\A$ and $\B$.

\begin{proposition}
$\A$ is equivalent to $\B$ if and only if $\A(w)=\B(w)$ for all
words $w \in L_{n^2}$, where $n$ is the sum of the number of states
of $\A$ and the number of states of $\B$.
\label{prop:equiv}
\end{proposition}
\begin{proof}
Recall that for each balanced word $u
\in \Sigma^*$ we have rational matrices $\Ms{\A}(u)$ and $\Ms{\B}(u)$
giving the respective state-to-state transition weights of $\A$
and $\B$ on reading $u$.  These two families of matrices can be
combined into a single family

\vspace{-4mm}
\[ \mathcal{M} = \left \{ \left(
                 \begin{array}{cc}
                  \Ms{\A}(u) & \mathbf{0}\\
                 \mathbf{0} & \Ms{\B}(u)
                 \end{array}\right)
: \mbox{$u$ well-matched} \right\}\] of $n \times n$ matrices.  Let us
                 also write $\mathcal{M}_i$ for the subset of
                 $\mathcal{M}$ generated by those well-matched words
                 $u \in L_i$.

Let $\alphas{\A},\etas{\A}$ and $\alphas{\B},\etas{\B}$ be the
respective initial and final-state vectors of $\A$ and $\B$.  Then
$\A$ is equivalent to $\B$ if and only if
\begin{gather}
( \begin{array}{cc} \alphas{\A} & \alphas{\B}
\end{array} ) M \left( \begin{array}{c}
                \etas{\A} \\ -\etas{\B}
\end{array} \right) = 0
\label{eq:zero}
\end{gather}
for all $M \in \mathcal{M}$.  It follows that $\A$ is equivalent to
$\B$ if and only if (\ref{eq:zero}) holds for all $M$ in
$\mathrm{span}(\mathcal{M})$, where the span is taken in the rational
vector space of $n \times n$ rational matrices.  But
$\mathrm{span}(\mathcal{M}_i)$ is an ascending sequence of vector spaces:
\[ \mathrm{Span}(\mathcal{M}_0) \subseteq
   \mathrm{Span}(\mathcal{M}_1) \subseteq
   \mathrm{Span}(\mathcal{M}_2) \subseteq \ldots
\]
It follows from a dimension argument that this sequence stops in at
most $n^2$ steps and we conclude that $\mathrm{span}(\mathcal{M})=
\mathrm{span}(\mathcal{M}_{n^2})$.\qed
\end{proof}

\begin{proposition}
  Given a weighted visibly pushdown automaton $\A$ and \mbox{$n
    \in\mathbb{N}$} one can compute in logarithmic space a circuit
  that represents $\sum_{w \in L_{n^2}} \A(w)$.
\label{prop:aut2circuit}
\end{proposition}
\begin{proof}
From the definition of the language $L_i$ and the family of
matrices $\Ms{\A}$ we have:
\begin{eqnarray*}
\sum_{w \in L_{i+1}} \Ms{\A}(w) & = & \sum_{a \in \Sigma_c}
                                    \sum_{b \in \Sigma_r}
                                    \sum_{\gamma \in \Gamma}
 \Ms{\A}(a,\gamma) \left( \sum_{u \in L_i} \Ms{\A}(u) \right) \Ms{\A}(b,\gamma)\\
 & & + \left(\sum_{u \in L_i} \Ms{\A}(u)\right)
       \left(\sum_{u \in L_i} \Ms{\A}(u)\right) \, .
\end{eqnarray*}
The above equation implies that we can compute in logarithmic space a
circuit that represents $\sum_{w \in L_n} \Ms{\A}(w)$.  The result of
the proposition immediately follows by premultiplying by the initial
state vector and postmultiplying by the final state vector.  \qed
\end{proof}

A key property of weighted VPA is their closure under product.
\newcommand{\stmtpropprod}{
Given weighted VPA $\A$ and $\B$ on the same alphabet $\Sigma$ one can
define a \emph{synchronous-product automaton}, denoted $\A \times \B$,
such that $(\A \times \B)(w) = \A(w) \B(w)$ for all $w \in \Sigma^*$.
}
\begin{proposition} \label{prop:prod}
\stmtpropprod
\end{proposition}

The proof of Proposition~\ref{prop:prod},
\iftechrep{given in Appendix~\ref{theAppendix},}{given in~\cite{KMOWW12:fossacs-report},}
exploits the fact that the stack height is
determined by the input word, so the respective stacks of $\A$ and
$\B$ operating in parallel can be simulated in a single stack.

\begin{proposition}
The equivalence problem for weighted visibly pushdown automata is
logspace reducible to \textbf{ACIT}.
\label{prop:red}
\end{proposition}
\begin{proof}
Let $\A$ and $\B$ be weighted visibly pushdown automata with a
total of $n$ states between them.  Then
\begin{eqnarray*}
 \sum_{w \in L_n} (\A(w)-\B(w))^2 & = &
 \sum_{w \in L_n} \A(w)^2 + \B(w)^2 - 2\A(w)\B(w) \\
& = &
\sum_{w \in L_n} (\A \times \A)(w) + (\B \times \B)(w) - 2(\A \times \B)(w)
\end{eqnarray*}
Thus $\A$ is equivalent to $\B$ iff \mbox{$\sum_{w \in L_n} (\A \times
  \A)(w)+ (\B \times \B)(w) = 2\sum_{w \in L_n} (\A \times \B)(w)$}.
But Propositions \ref{prop:aut2circuit} and~\ref{prop:prod} allow us
to translate the above equation into an instance of \textbf{ACIT}.
\qed
\end{proof}

The trick of considering sums-of-squares of acceptance weights in the
above proof is inspired by~\cite[Lemma 1]{Tzeng96}.

\bibliographystyle{plain} %or alpha or splncs
\bibliography{db}

\iftechrep{
\newpage
\appendix
\input{appendix1}
\input{appendix2}
}{}

\end{document}

%% file: appendix1.tex
\section{Proofs of Section~\ref{sec:commutative}}
\label{app:commutative}

%\begin{qproposition}{\ref{prop:converge}}
%\stmtpropconverge
%\end{qproposition}
%\begin{proof}
%Suppose that $\A$ is bounded.  Then for each $\vv \in \mathbb{Z}^s$ the
%matrix $M(\varepsilon)^m(\vv)$ is pointwise dominated by $M_b^m$ for
%all $m$.  But since $I+M_b+M_b^2+\cdots$ is a convergent series it
%follows that
%\begin{gather}
%I(\vv) + M(\varepsilon)(\vv) + M(\varepsilon)^2(\vv) + M(\varepsilon)^3(\vv) + \cdots
%\label{eq:series}
%\end{gather}
%is also a convergent series for each $\vv \in \mathbb{Z}^s$.
%\qed
%\end{proof}

\begin{qproposition}{\ref{prop:shortword}}
\stmtpropshortword
\end{qproposition}
\begin{proof}
  Suppose that $\A$ is non-zero.  Then there exists some word $w \in
  \Sigma^*$ such that $\valpha M(w) \veta$ is a non-zero rational
  expression in $\vx$.  Thus we can pick $\vr \in \Q^s$ such that
  $\valpha M(w)(\vr) \veta \neq 0$.  Now define
\[ V_i := \mathop{\mathrm{Span}}
\{ \valpha M(w)(\vr) \in \Q^n : |u| \leq i \} \, .\]

Then $V_0 \subseteq V_1 \subseteq V_2 \subseteq \ldots$ is an
increasing family of subspaces of $\Q^n$.  Considering the dimension
of each $V_i$ there exists $i_0 < n$ with $V_{i_0}=V_{i_0+1}$.  But for
all $i$ we have
\[
V_{i+1}  =  \mathrm{Span}(V_i \cup
\{(\vv M(a)M(\varepsilon)^*)(\vr) : a \in \Sigma, \vv \in V_i \})
\]
From this characterisation it is clear that $V_{i_0+1}=V_{i_0}$ entails that
$V_i = V_{i_0}$ for all $i \geq i_0$.

From the assumption that $\A(w)\not\equiv 0$ as a rational function
there exists $\vr \in \Q^s$ such that $\A(w)(\vr) \neq 0$.  By the
above we have that $\valpha M(w)(\vr) \in V_{i_0}$ and thus $\veta$ is
not orthogonal to $V_{i_0}$.  In particular, there exists some word
$u$ of length at most $i_0 \leq n-1$ such that $\valpha M(u)(\vr) \eta
\neq 0$.  But then $\A(u) \not\equiv 0$.
\qed
\end{proof}

\begin{qcorollary}{\ref{cor:counter}}
\stmtcorcounter
\end{qcorollary}
\begin{proof}
Consider a counter automaton $\A$ to be checked for zeroness.  If the
number $s$ of counters is fixed, then the set of sample points $R^s$
in the proof of Theorem~\ref{thm:counter} has polynomial size.  Thus
we can in polynomial time test whether $\A(w)(\vr)=0$ for all $w \in
\Sigma^*$ and $\vr \in R^s$.  If $\A$ is non-zero then the proof of
Theorem~\ref{thm:counter} guarantees the existence of $w \in
\Sigma^*$ and $\vr \in R^s$ such that $\A(w)(\vr)\neq 0$.
\qed
\end{proof}

%% file: appendix2.tex
\section{Proofs of Section~\ref{sec:vpa}}
\label{theAppendix}
\begin{qproposition}{\ref{prop:ACIT-to-VPA}}
\stmtpropACITtoVPA
\end{qproposition}
\begin{proof}
Let $C$ and $C'$ be two circuits over basis $\{+,*\}$.  Without loss
of generality we assume that in each circuit the inputs of a
depth-$i$ gate both have depth $i+1$, $+$-nodes have even depth,
$*$-nodes have odd depth, and input nodes all have the same depth $d$.
Notice that in either circuit any path from an input gate to an
output gate has length $d$.

We define two automata $\mathcal{A}$ and $\mathcal{A}'$ that are
equivalent if and only if $C$ and $C'$ have the same output.  Both
automata are defined over the alphabet $\{c,r,\iota\}$, with $c$ a
call, $r$ a return and $\iota$ an internal event.  We explain how
$\mathcal{A}$ arises from $C$; the definition of $\mathcal{A}'$ is
entirely analogous.

Suppose that $C$ has set of gates $\{g_0,g_1,\ldots,g_n\}$, with $g_0$
the output gate.  For each gate $g_i$ of $C$ we include a state $s_i$
of $\mathcal{A}$ and a stack symbol $\gamma_i$.  The initial state of
$\mathcal{A}$ is $s_0$, and all states are accepting.  The transitions
of $\mathcal{A}$ are defined as follows:

\begin{itemize}
\item For each $+$-gate $g_i:=g_j+g_k$ in $C$ we include an internal
  transition from $s_i$ that goes to $s_j$ with probability $1/2$ and
  to $s_k$ with probability $1/2$.
\item For each $*$-gate $g_i := g_j * g_k$ we include a
  probability-$1$ call transition from $s_i$ to $s_j$ that pushes
  $\gamma_k$ onto the stack.
\item
An input gate $g_i$ with label $0$ contributes no transitions.
\item For each input gate $g_i$ with label $1$ and each stack symbol
  $\gamma_j$, we include a return transition from $s_i$ that pops
  $\gamma_j$ off the stack and ends in state $s_j$ with probability
  $1$.
\end{itemize}
Recall that acceptance is by empty stack and final state.  By
construction $\mathcal{A}$ only accepts a single word, as we now
explain.  Define a sequence of words $w_n \in \{c,r,\iota\}^*$ by $w_0
= \iota$, $w_{n+1} = \iota w_n$ for $n$ even, and
$w_{n+1} =  c w_n r w_n$ for $n$ odd.
%\[ w_{n+1} = \left\{ \begin{array}{ll}
%         \iota w_n & \mbox{$n$ even}\\
%         c w_n r w_n & \mbox{$n$ odd.}
%\end{array}\right . \]
Furthermore, write $M_0 = 1$, $M_{n+1} = 2M_n$ for $n$ even, and
$M_{n+1} = M_n^2$ for $n$ odd.  Then $\mathcal{A}$ accepts $w_d$ with
probability $N/M_d$, where $d$ is the depth of the circuit $C$ and $N$
is output of $C$.  All other words are accepted with probability $0$.
We conclude that $C$ and $C'$ have the same value if and only if
$\mathcal{A}$ and $\mathcal{A}'$ are equivalent.
\end{proof}

\begin{qproposition}{\ref{prop:prod}}
\stmtpropprod
\end{qproposition}
\begin{proof}
Let $\A=(\ns{\A},\Sigma,\Gammas{\A},\Ms{\A},\alphas{\A},\etas{\A})$
and $\B=(\ns{\B},\Sigma,\Gammas{\A},\Ms{\B},\alphas{\B},\etas{\B})$.
We define a product automaton $\C$.  Note that since the stack height
is determined by the input word we can simulate the respective stacks
of $\A$ and $\B$ using a single stack in $\C$ whose alphabet is the
product of the respective stack alphabets of $\A$ and $\B$.

The number of states of $\C$ is $\ns{\A} \cdot \ns{\B}$.  The initial
vector $\alphas{\C}$ has $(i,j)$-th component
$\alphas{\A}_i \cdot \alphas{\B}_j$.  The final vector $\etas{\C}$ is
defined likewise.  The stack alphabet of $\C$ is
$\Gammas{\A} \times \Gammas{\B}$.  Given
$a \in \Sigma_c \cup \Sigma_r$ we define the $((i,j),(k,l))$-th
component of the transition matrix $\Ms{\C}(a,(\gamma,\gamma'))$ to be
the product of $\Ms{\A}(a,\gamma)$ and $\Ms{\B}(a,\gamma')$.
\qed
\end{proof}